\documentclass{article}
\date{
November  7, 2013%
}

\usepackage[T1]{fontenc}

\advance\textwidth by 2cm
\advance\oddsidemargin by -1cm
\advance\evensidemargin by -1cm

\advance\textheight by 2cm
\advance\topmargin by -1cm
 
\usepackage[latin1]{inputenc}
\usepackage{pdfsync}
\usepackage{url}
\usepackage{graphicx}
\usepackage{amsfonts,amsmath,mathrsfs,theorem,bm}
\usepackage{rotating}
\usepackage{algorithm, algorithmic}
\floatname{algorithm}{Algorithm}
\usepackage[usenames,dvipsnames]{color}
\usepackage{multirow}
\usepackage{dcolumn}
\usepackage{colortbl}

\newtheorem{definition}{Definition}
\newtheorem{theorem}{Theorem}
\newcommand\fakepar[1]{\smallskip\noindent\textbf{#1.}}

\newcounter{noqed}
\newcommand{\qed}{ \ifmmode\mbox{ }\fi\rule[-.05em]{.3em}{.7em}\setcounter{noqed}{0}}
\newenvironment{proof}[1][{}]{\noindent{\bf Proof#1. }\setcounter{noqed}{1}}{\ifnum\value{noqed}=1\qed\fi\par\medskip}

\renewcommand{\epsilon}{\varepsilon}

\renewcommand{\theta}{\vartheta}
\newcommand{\1}{\mathbf 1}
\newcommand{\lst}[2]{${#1}_0$,~${#1}_1$, $\dots\,$,~${#1}_{#2-1}$}

\def\..{\,\mathpunct{\ldotp\ldotp}} 

\newcommand{\comment}[1]{}

\title{Axioms for Centrality}
\author{Paolo Boldi\quad Sebastiano Vigna\thanks{The authors have been supported by the
EU-FET grant NADINE (GA 288956).}\\\normalsize Dipartimento di informatica, Universit\`a degli Studi di Milano, Italy}

\begin{document}
\bibliographystyle{plain}
\maketitle

\begin{abstract}
Given a social network, which of its nodes are more central? This question 
has been asked many times in sociology, psychology and computer science, and a
whole plethora of \emph{centrality measures} (a.k.a.~\emph{centrality indices}, or \emph{rankings}) were proposed to account for
the importance of the nodes of a network. In this paper, we try to provide a 
mathematically sound survey of the most important classic centrality measures
known from the literature and propose an \emph{axiomatic} approach 
to establish whether they are actually doing what they have been
designed for. Our axioms suggest
some simple, basic properties that a centrality measure should exhibit.

Surprisingly, only a new simple measure based on distances,
\emph{harmonic centrality}, turns out to satisfy all axioms; essentially, harmonic
centrality is a correction to Bavelas's classic \emph{closeness
centrality}~\cite{BavMMGS} designed to take unreachable nodes into account in a natural way. 

As a sanity check, we examine in turn each measure under the lens of information
retrieval, leveraging state-of-the-art knowledge in the discipline to measure
the effectiveness of the various indices in locating web pages that are relevant
to a query. While there are some examples of such comparisons in the literature, here
for the first time we also take into consideration centrality
measures based on distances, such as closeness, in an information-retrieval
setting. The results closely match the data we gathered using our axiomatic approach.

Our results suggest that centrality measures based on distances, which in
the last years have been
neglected in information retrieval in favor of spectral centrality measures, 
do provide high-quality signals; moreover, harmonic centrality
pops up as an excellent general-purpose centrality index for arbitrary directed graphs.
\end{abstract}

\section{Introduction}

In recent years, there has been an ever-increasing research activity in the
study of real-world complex networks~\cite{WFSNAMA} (the
world-wide web, the autonomous-systems graph within the Internet, coauthorship graphs,
phone-call graphs, email graphs and biological networks, to cite but a few). 
These networks, typically generated directly or indirectly by human activity and
interaction (and therefore hereafter dubbed ``social''), appear in a large variety of contexts and often exhibit a surprisingly similar 
structure. One of the most important notions that researchers have been trying to capture
in such networks is ``node centrality'':
ideally, every node (often representing an
individual) has some degree of influence or importance within the social domain
under consideration, and one expects such importance to surface in the
structure of the social network; centrality is a quantitative measure that
aims at revealing the importance of a node.

Among the types of centrality that have been considered in the literature
(see~\cite{BorCNF} for a good survey), many have to do with distances
between nodes.\footnote{Here and in the
following, by ``distance'' we mean the length of a shortest path between two
nodes.} Take, for instance, a node in an undirected connected network: 
if the sum of distances to all other nodes is large, the node under consideration 
is \emph{peripheral}; this is
the starting point to define Bavelas's \emph{closeness centrality}~\cite{BavMMGS}, which is the
reciprocal of peripherality (i.e., the reciprocal of the sum of distances to all other
nodes).

The role played by shortest paths is justified by one of the most well-known
features of complex networks, the so-called \emph{small-world} phenomenon.
A small-world network~\cite{CHCNSRF} is a graph where the average distance
between nodes is logarithmic in the size of the network, whereas the clustering
coefficient is larger (that is, neighborhoods tend to be denser) than in a
random Erd\H os-R\'enyi graph with the same size and average
distance.\footnote{The reader might find this definition a bit vague, and some
variants are often spotted in the literature: this is a general well-known problem,
also highlighted recently, for example in~\cite{LADTTSFG}.} The fact that social networks 
(whether electronically mediated or not)
exhibit the small-world property is known at least since Milgram's famous
experiment~\cite{MilSWP}
and is arguably the most popular of all features of complex networks. For instance,
the average distance of the Facebook graph was recently established to be just
$4.74$~\cite{BBRFDS}. 


The purpose of this paper is to pave the way for a formal well-grounded
assessment of centrality measures, based on some simple guiding principles; we
seek notions of centrality that are at the same time \emph{robust} (they should
be applicable to arbitrary directed graphs, possibly non-connected,
without modifications) and \emph{understandable} (they should have a clear combinatorial
interpretation).

With these principles in mind, we shall present and compare the most popular and
well-known centrality measures proposed in the last decades. The comparison will be based
on a set of \emph{axioms}, each trying to capture a specific trait.

In the last part of the paper, as a sanity check, we compare the measures we discuss
in an information-retrieval setting, extracting from the classic GOV2 web collection
documents satisfying a query and ranking by centrality the 
subgraph of retrieved documents.

The results are somehow surprising, and
suggest that simple measures based on distances, and in 
particular \emph{harmonic centrality} (which we
introduce formally in this paper) can give better results than some of the most
sophisticated indices used in the literature. These unexpected outcomes are the
main contribution of this paper, together with the set of axiom we propose,
which provide a conceptual framework for understanding centrality measures in a
formal way.  We also try to give an orderly account of centrality in
social and network sciences, gathering scattered results and folklore knowledge
in a systematic way.

\section{A Historical Account}

In this section, we sketch the historical development of centrality, focusing on 
ten classical centrality measures that we decided to include in this paper: 
the overall growth of the field is of course much more complex, and the literature 
contains a myriad of alternative proposals that will not be discussed here.

Centrality is a fundamental tool in the study of social networks:
the first efforts to define formally centrality indices were put forth in the
late 1940s by the Group Networks Laboratory at MIT~directed by Alex
Bavelas~\cite{BavMMGS}, in the framework of communication patterns and
group collaboration~\cite{LeaSECCPGP,BBEAOC}; those pioneering experiments
concluded that centrality was related to group efficiency in problem-solving, 
and agreed with the subjects' perception of leadership.
In the following decades, various measures of centrality were employed in
a multitude of contexts (to understand political integration in Indian social
life~\cite{CMNCIIC}, to examine the consequences of centrality in communication
paths for urban development~\cite{PitGTAHG}, to analyze their implications to
the efficient design of organizations~\cite{BeaIIC,MacSCCN}, or even to explain
the wealth of the Medici family based on their
central position with respect to marriages and financial
transactions in the 15th century Florence~\cite{PARARM}).
We can certainly say that the problem of singling out influential individuals
in a social group is a holy grail that sociologists have been trying to capture for at 
least sixty years. 

Although all researchers agree that centrality is an
important structural attribute of social networks and that it is directly related to
other important group properties and processes, there is no consensus on \emph{what} 
centrality 
is exactly or on its conceptual foundations, and there is very little agreement
on the proper procedures for its measurement~\cite{BurCNBC,FreCSNCC}.
It is true that often different centrality indices are designed to capture different
properties of the network under observation, but 
as Freeman observed, ``several measures are often only vaguely related to
the intuitive ideas they purport to index, and many are so complex that it is difficult 
or impossible to discover what, if anything, they are measuring''~\cite{FreCSNCC}.

Freeman acutely remarks that the implicit starting point of all centrality
measures is the same: the central node of a star should be deemed more important
than the other vertices; paradoxically, it is precisely the unanimous agreement
on this requirement that may have produced quite different approaches to the
problem. In fact, the center of a star is at the same time
\begin{enumerate}
  \item the node with largest degree;
  \item the node that is closest to the other nodes (e.g., that has the smallest average distance to other nodes);
  \item the node through which most shortest paths pass;
  \item the node with the largest number of incoming paths of length $k$, for every $k$;
  \item the node that maximizes the dominant eigenvector of the graph matrix;
  \item the node with the highest probability in the stationary distribution of the natural random walk on the graph.
\end{enumerate}
These observations lead to corresponding (competing) views of centrality. Degree is
probably the oldest measure
of importance ever used, being equivalent to majority voting in elections 
(where $x\to y$ is interpreted as ``$x$ voted for $y$'').

The most classical notion of \emph{closeness}, instead, was introduced by
Bavelas~\cite{BavMMGS} for undirected, connected networks as the reciprocal of
the sum of distances from a given node. Closeness was originally aimed at
establishing how much a vertex can communicate without relying on third parties
for his messages to be delivered.\footnote{The notion can also be generalized to
a weighted summation of node contributions multiplied by some \emph{discount} functions
applied to their distance to a given node~\cite{CoKSDAN}.} In the seventies, Nan
Lin proposed to adjust the definition of closeness so to make it usable on directed networks that
are not necessarily strongly connected~\cite{LinFSR}.


Centrality indices based on the count of shortest paths were formally developed
independently by Anthonisse~\cite{AntRG} and Freeman~\cite{FreSMCBB}, who
introduced \emph{betweenness} as a measure of the probability that a random shortest
path passes through a given node or edge.

Katz's index~\cite{KatNSIDSA} is based instead on a weighted count of 
\emph{all} paths coming into a node: more precisely, the weight of a path
of length $t$ is $\beta^t$, for some \emph{attenuation factor} $\beta$,
and the score of $x$ is the sum of the weights of all paths coming into $x$.
Of course, $\beta$ must be chosen so that all the summations converge.  



While the above notions of centrality are combinatorial in nature and based on
the discrete structure of the underlying graph,
another line of research studies \emph{spectral} techniques (in the sense of linear
algebra) to define centrality~\cite{VigSR}.

The earliest known proposal of this kind is due to Seeley~\cite{SeeNRI}, who normalized
the rows of an adjacency matrix representing the ``I like him'' relations
among a group of children, and assigned a centrality score using the resulting left dominant
eigenvector. This approach is equivalent to studying the stationary distribution
of the Markov chain defined by the natural random walk on the graph. Few years later, 
Wei~\cite{WeiAFRT} proposed to rank sport teams using the right dominant
eigenvector of a tournament matrix, which contains 1 or 0 depending on whether a team
defeated another team.
Wei's work was then popularized by Kendall~\cite{KenFCTPC}, and the technique is 
known in the literature about ranking of sport teams as ``Kendall--Wei
ranking''.  Building on Wei's approach, Berge~\cite{BerTGA}
generalized the usage of dominant eigenvectors to arbitrary directed graphs, and in particular
to \emph{sociograms}, in which an arc represents a relationship of influence between two individuals.\footnote{Dominant eigenvectors were rediscovered as a generic way of computing centralities on graphs by
Bonacich~\cite{BonFWASSCI}.}

Curiously enough, the most famous among spectral centrality scores is also 
one of the most recent, PageRank~\cite{PBMPCR}:
%
PageRank was a centrality measure specifically geared toward web graphs, and it
was introduced precisely with the aim of implementing it in a search engine
(specifically, Google, that the inventors of PageRank founded in 1997). 

In the same span of years, Jon Kleinberg defined another centrality measure
called HITS~\cite{KleASHE} (for
``Hyperlink-Induced Topic Search''). The idea\footnote{To be precise, Kleinberg's
algorithm works in two phases; in the first phase, one selects a subgraph of the
starting webgraph based on the pages that match the given query; in the second
phase, the centrality score is computed on the subgraph. Since in this paper we
are looking at HITS simply as a centrality index, we will simply apply it to the
graph under examination.} is that every node of a graph is associated with two
importance indices: one (called ``authority score'') measures how reliable
(important, authoritative\dots) a node is, and another (called ``hub score'')
measures how good the node is in pointing to authoritative nodes, with the two
scores mutually reinforcing each other. The result is again the dominant
eigenvector of a suitable matrix.
SALSA~\cite{LeMS} is a more recent and strictly related score based on the same
idea, with the difference that it applies some normalization to the matrix.

We conclude this brief historical account mentioning that there were in the past some attempts to 
axiomatize the notion of centrality: we postpone a discussion on these attempts
to Section~\ref{sec:histax}.

\section{Definitions and conventions}
\label{sec:defs}

In this paper, we consider directed graphs defined by a set $N$ of $n$ nodes and a set
$A\subseteq N\times N$ of arcs; we write $x\to y$ when $\langle x,y\rangle \in A$ and call
$x$ and $y$ the source and target of the arc, respectively. 
An arc with the same source and target is called a \emph{loop}. 

The \emph{transpose} of a graph is obtained by reversing all arc directions (i.e., it has an arc $y\to x$
for every arc $x\to y$ of the original graph).
A \emph{symmetric graph} is a graph such that $x\to y$ whenever $y\to x$; such a
graph is fixed by transposition, and can be identified with an undirected graph, that is, a graph whose arcs (usually called \emph{edges}) are a
subset of unordered pairs of nodes. 
A \emph{successor} of $x$ is a node $y$ such that $x\to y$, and a \emph{predecessor} of $x$ is a node $y$ such that $y\to x$.
The \emph{outdegree} $d^+(x)$ of a node $x$ is the number of its successors, and the
\emph{indegree} $d^-(x)$ is the number of its predecessors.

A \emph{path} (of length $k$) is a sequence \lst xk, where $x_j\to x_{j+1}$,
$0\leq j <k$.
A \emph{walk} (of length $k$) is a sequence \lst xk, where $x_j\to x_{j+1}$ or
$x_{j+1}\to x_j$, $0\leq j <k$.
A connected (strongly connected, respectively) \emph{component} of a graph is a maximal subset in which
every pair of nodes is connected by a walk (path, respectively). Components form a partition of the nodes of a graph.
A graph is \emph{(strongly) connected} if there is a single (strongly) connected component,
that is, for every choice of $x$ and $y$ there is a walk (path) from $x$ to $y$. A strongly connected
component is \emph{terminal} if its nodes have no arc towards other components.

The \emph{distance} $d(x,y)$ from $x$ to $y$ is the length of a shortest path from $x$ to $y$,
or $\infty$ if no such path exists. The nodes \emph{reachable} from $x$ are the nodes $y$ such that $d(x,y)<\infty$.
The nodes \emph{coreachable} from $x$ are the nodes $y$ such that $d(y,x)<\infty$. A node
has \emph{trivial} (co)reachable set if the latter contains only the node itself.

The notation $\bar A$, where $A$ is a nonnegative matrix, will be used throughout the paper to denote the matrix obtained by
$\ell_1$-normalizing the rows of $A$, that is, dividing each element of a row by the sum of the row (null rows are left unchanged).
If there are no null rows, $\bar A$ is (row-)\emph{stochastic}, that is, it is nonnegative and the row sums are all equal to one.

We use Iverson's notation: if $P$ is a predicate, $[P]$ has value $0$ if $P$ is false and $1$ if $P$ is true~\cite{KnuTNN};
we denote with $H_i$ the $i$-th harmonic number $\sum_{1\leq k\leq i}1/k$; finally, $\bm \chi_i(j)=[j=i]$ is the characteristic
vector of $i$. We number graph nodes and corresponding vector coordinates starting from zero.

%
%
%

\subsection{Geometric measures}

We call \emph{geometric} those measures assuming that
importance is a function of distances; more precisely, a geometric centrality depends only on how many nodes
exist at every distance. These are some
of the oldest measures defined in the literature.  

\subsubsection{Indegree}
Indegree, the number of incoming arcs $d^-(x)$, can be considered a geometric measure: it
is simply the number of nodes at distance one\footnote{Most centrality measures proposed in the literature were
actually described only for undirected, connected graphs. Since the study of
web graphs and online social networks has posed the problem of extending
centrality concepts to networks that are directed, and possibly not strongly connected,
in the rest of this paper we consider measures depending on the \emph{incoming}
arcs of a node (e.g., incoming paths, left dominant eigenvectors, distances from all nodes
to a fixed node). If necessary, these measures can be called ``negative'', as opposed
to the ``positive'' versions obtained by considering outgoing paths, or (equivalently) by transposing the graph.}. It is probably the oldest measure
of importance ever used, as it is equivalent to majority voting in elections (where $x\to y$ if $x$ voted for $y$).
Indegree has a number of obvious shortcomings (e.g., it is easy to spam), but it is a good baseline,
and in some cases turned out to provide better results than more sophisticated methods (see, e.g.,~\cite{CHUPFFPI}).

\subsubsection{Closeness}
Bavelas introduced closeness in the late forties~\cite{BavCPTOG}; the closeness of $x$ is defined by 
\begin{equation}
\label{eq:closeness}
	 \frac1{\sum_y d(y,x)}.
\end{equation}
The intuition behind closeness is that nodes that are more central have smaller distances, and thus a smaller
denominator, resulting in a larger centrality.
We remark that for this definition to make sense, the graph must be strongly
connected. Lacking that condition, some of the denominators will be $\infty$, 
resulting in a null score for all nodes that cannot coreach the whole graph.

It was probably not in Bavelas's intentions to apply the measure to directed graphs, and even less to 
graphs with infinite distances,
but nonetheless closeness is sometimes ``patched'' by simply not including unreachable nodes, that is,
\[
\frac1{\sum_{d(y,x)<\infty} d(y,x)},
\]
and assuming that nodes with an empty coreachable set have centrality $0$ by definition: this is actually
the definition we shall use in the rest of the paper.
These apparently innocuous adjustments, however, introduce a strong bias toward nodes with a small coreachable
set.

\subsubsection{Lin's index}
Nan Lin~\cite{LinFSR} tried to repair the definition of closeness for graphs
with infinite distances by weighting closeness using the
square of the number of coreachable nodes; his definition for the centrality of a node $x$ with a nonempty coreachable set is
\[
	\frac{\bigl|\{y\mid d(y,x)<\infty\}\bigr|^2}{\sum_{d(y,x)<\infty} d(y,x)}.
\]
The rationale behind this definition is the following: first, we consider closeness not the inverse
of a sum of distances, but rather the inverse of the \emph{average} distance, which entails a first
multiplication by the number of coreachable nodes. This change normalizes closeness across the
graph. Now, however, we want nodes with a larger coreachable set to be more important, given
that the average distance is the same, so we multiply again by the number of coreachable nodes.
Nodes with an empty coreachable set have centrality $1$ by definition.

Lin's index was (somewhat surprisingly) ignored in the following literature. Nonetheless, it
seems to provide a reasonable solution for the problems caused by the definition of closeness.

\subsubsection{Harmonic centrality}
As we noticed, the main problem of closeness lies in the presence of pairs of unreachable nodes. 
We thus get inspiration from Marchiori
and Latora~\cite{MaLHSW}: faced with the problem of providing a sensible
notion of ``average shortest path'' for a generic directed network, they propose to
replace the average distance with the \emph{harmonic mean of all distances} (i.e., the $n(n-1)$ distances
between every pair of distinct nodes).
Indeed, in case a large number of pairs of nodes are not reachable, the average
of finite distances can be misleading: a graph might have a very
low average distance while it is almost completely disconnected (e.g., a
perfect matching has average distance $1/2$). The harmonic mean has the
useful property of handling $\infty$ cleanly (assuming, of course, that
$\infty^{-1}=0$).
For example, the harmonic mean of distances of a perfect matching is $n-1$: in fact, for every node there is
exactly another node at a non-infinite distance, and its distance is 1; so the sum of the inverse of all distances
is $n$, making the harmonic average equal to $n(n-1)/n=n-1$.

In general, for each graph-theoretical notion based on arithmetic averaging or
maximization there is an equivalent notion based on the harmonic mean. If we consider closeness the reciprocal of a
denormalized average of distances, it is natural to consider also the reciprocal of a
denormalized harmonic mean of distances.
We thus define the \emph{harmonic
centrality} of $x$ as\footnote{We remark that Yannick Rochat, in a talk at ASNA 2009 (Application of Social
Network Analysys),  
observed that in an undirected graph with several disconnected components the inverse of the harmonic mean of distances
offers a better notion of centrality than closeness, as it weights less those elements that belong to
smaller components. Tore Opsahl made the same observation in a March 2010 blog posting. 
Raj Kumar Pan and Jari Saram\"aki deviated from the classical definition of closeness in~\cite{PaSPLCCTN},
using in practice harmonic centrality, with the motivation of better handling disconnected nodes (albeit with
no reference to the harmonic mean). Edith Cohen and Haim Kaplan defined in~\cite{CoKSDAN} the notion
of a \emph{spatially decaying aggregate}, of which harmonic centrality is a particular instance.} 
\begin{equation}
\label{eq:harmonic}
	\sum_{y\neq x}\frac1{d(y,x)} = \sum_{d(y,x)<\infty, y\neq x}\frac1{d(y,x)}.
\end{equation}
The difference with~(\ref{eq:closeness}) might seem minor, but actually it
is a radical change. Harmonic centrality is strongly correlated to closeness centrality in simple networks, but naturally also accounts for 
nodes $y$ that cannot reach
$x$. Thus, it can be fruitfully applied to graphs that are not strongly connected. 

\subsection{Spectral measures}

\emph{Spectral measures} compute the left dominant eigenvector of some matrix derived from the graph, and
depending on how the matrix is modified before the computation we can obtain a number of different
measures. Existence and uniqueness of such measures is usually derivable by the theory of
nonnegative matrices started by Perron and Frobenius~\cite{BePNMMS}; we will not, however,
discuss such issues, as there is a large body of established literature about the topic. All 
observations in this section are true for strongly connected graphs; the modifications for
graphs that are not strongly connected can be found in the cited references.

\subsubsection{The left dominant eigenvector}
The first and most obvious spectral measure is the left dominant eigenvector of the plain adjacency
matrix. Indeed, the dominant eigenvector can be thought as the fixed point of an iterated
computation in which every node starts with the same score, and then replaces its score
with the sum of the scores of its predecessors. The vector is then normalized, and the process repeated
until convergence.

Dominant eigenvectors fail to behave as expected on graphs that are not strongly connected. Depending on the
dominant eigenvalue of the strongly connected components, the dominant eigenvector might
or might not be nonzero on non-terminal components (a detailed characterization can be found in~\cite{BePNMMS}).

\subsubsection{Seeley's index}
The dominant eigenvector rationale can be slightly amended with the observation that
the update rule we described can be thought of as if each node gives away its score
to its successors: or even that each node has a \emph{reputation} and is giving
its reputation to its successors so that they can build their own. 

Once we take this viewpoint, it is clear that it is not very
sensible to give away our own amount of reputation to everybody: it is more reasonable
to \emph{divide} equally our reputation among our successors. From a linear-algebra
viewpoint, this corresponds to normalizing each row of the adjacency matrix using the $\ell_1$ norm.

Seeley advocated this approach~\cite{SeeNRI} for computing the popularity among groups
of children, given a graph representing whether each child liked or not another one.
The matrix resulting from the $\ell_1$-normalization process is stochastic,
so the score can be interpreted as the stationary state of a Markov
chain. In particular, if the underlying graph is symmetric then Seeley's index collapses 
to degree (modulo normalization) because of the well-known characterization of the
stationary distribution of the natural random walk on a symmetric graph.

Also Seeley's index does not react very well to the lack of strong connectivity:
the only nodes with a nonzero score are those belonging to terminal components
that are not formed by a single node of outdegree zero.

\subsubsection{Katz's index}
Katz introduced his celebrated index~\cite{KatNSIDSA} using a summation over all
paths coming into a node, but weighting each path so that the summation would 
be finite. Katz's index can be expressed as
\[
\bm k = \mathbf 1\sum_{i=0}^\infty \beta^iA^i,
\]
due to the interplay between the powers of the adjacency matrix and the number of
paths connecting two nodes.
For the summation above to be finite, the \emph{attenuation factor} $\beta$ 
must be smaller than $1/\lambda$, where $\lambda$ is the dominant eigenvalue of $A$.

Katz immediately noted that the index was expressible using linear algebra
operations:
\[ \bm k = \mathbf 1( 1 - \beta A ) ^ {-1} .
\] It took some more time to realize that, due to Brauer's theorem on the
displacement of eigenvalues~\cite{BraLCRMIV}, Katz's index is the left dominant
eigenvector of a \emph{perturbed matrix}
\begin{equation}
\label{eq:katzpert}
\beta\lambda A + ( 1 - \beta\lambda ) \bm e^T \mathbf 1,
\end{equation}
where $\bm e$ is a right dominant eigenvector of $A$ such that $\mathbf 1\bm e^T=\lambda$~\cite{VigSR}. An easy generalization
(actually suggested by Hubbell~\cite{HubIOACI}) replaces the vector $\mathbf 1$ with 
some preference vector $\bm v$ so that paths are also weighted differently depending on their
starting node.\footnote{We must note that the original definition of Katz's index is 
$\mathbf 1A\sum_{i=0}^\infty \beta^iA^i=\mathbf 1/\beta\sum_{i=0}^\infty \beta^{i+1}A^{i+1}=
(\mathbf 1/\beta)\sum_{i=0}^\infty \beta^iA^i - \mathbf 1/\beta$. This additional multiplication by
$A$ is somewhat common in the literature, even for PageRank; clearly, it alters the order induced by the scores
only when there is a nonuniform preference vector. Our discussion 
can be easily adapted for this version.} 

The normalized limit of Katz's index when $\beta\to 1/\lambda$ is a dominant eigenvector;
if the dominant eigenvector is not unique, the limit depends on the preference vector $\bm v$~\cite{VigSR}.


\subsubsection{PageRank}
PageRank is one of the most discussed and quoted spectral indices in use today, mainly because of
its alleged use in Google's ranking algorithm.\footnote{The reader should be aware, however, that the literature about the actual 
effectiveness of PageRank in information retrieval is rather scarce, and comprises mainly negative results such 
as~\cite{NZTHW} and~\cite{CHUPFFPI}.}
By definition, PageRank is the unique vector $\bm p$ satisfying 
\begin{equation}
\label{eq:prlin}
\bm p = \alpha \bm p \bar A + (1-\alpha)\bm v,
\end{equation}
where $\bar A$ is again the $\ell_1$-normalized adjacency matrix of the graph, $\alpha\in[0\..1)$ is a
\emph{damping factor}, and $\bm v$ is a \emph{preference
vector} (which must be a distribution).
This is the definition appearing in Brin and Page's paper about Google~\cite{BrPALHWSE}; 
the authors claim that the PageRank score is a probability distribution
over web pages, that is, it has unit $\ell_1$ norm, but this is not necessarily true if $A$ has null rows. 
The following scientific literature
thus proposed several ways to \emph{patch} the matrix $\bar A$ to make it stochastic, which would guarantee
$\|\bm p\|_1=1$. 
A common solution is to replace every null row with the preference vector $\bm v$ itself,
but other solutions have been proposed (e.g., adding a loop to all nodes of outdegree zero), leading to different scores.
This issue is definitely not academic, as in typical web snapshots a significant fraction of the nodes have
outdegree zero (the so-called \emph{dangling nodes}).

It is interested to note, however, that in the preprint
written in collaboration with Motwani and Winograd and commonly quoted as defining PageRank~\cite{PBMPCR}, Brin and Page themselves propose a different 
but essentially equivalent linear recurrence in the style of Hubbell's index~\cite{HubIOACI},
and acknowledge that $\bar A$ can have null rows, in which case the dominant eigenvalue of $\bar A$
could be smaller than one, and the solution might need to be normalized to have unit $\ell_1$ norm. 

Equation~(\ref{eq:prlin}) is of course solvable even without any patching, giving
\begin{equation}
\label{eq:prinv}
\bm p = (1-\alpha) \bm v\bigl( 1 - \alpha \bar A \bigr) ^ {-1} ,
\end{equation}
and finally
\begin{equation}
\label{eq:prsum}
\bm p = (1-\alpha) \bm v\sum_{i=0}^\infty \alpha^i\bar A^i,
\end{equation}
which shows immediately that Katz's index and PageRank differ only by a constant factor and by the $\ell_1$ normalization applied
to $A$, similarly to the difference between the dominant eigenvector and Seeley's index.

If $A$ has no null rows, or $\bar A$ has been patched to be stochastic, PageRank can be equivalently
defined as the stationary distribution (i.e., the left dominant eigenvector) of the Markov chain with
transition matrix
\begin{equation}
\label{eq:prmarkov}
\alpha \bar A + ( 1 - \alpha ) \1^T \bm v,
\end{equation}
which is of course analogous to~(\ref{eq:katzpert}).
Del Corso, Gull\`i and Romani~\cite{DCGRFPCSLS} have shown that
the resulting scores (which have always unit $\ell_1$ norm) differ from the PageRank vector
defined by~(\ref{eq:prlin}) only by a normalization factor, provided that in~(\ref{eq:prmarkov}) null rows have been replaced with $\bm v$.
If $A$ had no null rows, the scores are of course identical as~(\ref{eq:prinv}) can be easily derived from~(\ref{eq:prmarkov}). 

Both definitions have been used in the literature: the linear-recurrence definition~(\ref{eq:prlin})
is particularly useful when one needs linear dependence on $\bm v$~\cite{BSVPFDF}. The Markov-chain
definition~(\ref{eq:prmarkov}) is nonetheless more common, albeit it raises the issue of patching null rows.
  
Analogously to what happens with Katz's index, if the dominant eigenvalue of
$\bar A$ is one then the limit of PageRank when $\alpha$ goes to 1 is a dominant
eigenvector of $\bar A$, that is, Seeley's index, and if the dominant
eigenvector is not unique the limit depends on the preference vector $\bm
v$~\cite{BSVPFD}.

In the rest of this paper, except when explicitly stated, we shall define PageRank as in equation (\ref{eq:prinv}), 
without any patching or normalization, and use a uniform preference vector.

\subsubsection{HITS}
Kleinberg introduced his celebrated HITS algorithm~\cite{KleASHE}\footnote{A few years before, Bonacich~\cite{BonSGIC} had introduced
the same method for computing the centrality of individuals and groups related by a
rectangular incidence matrix. Kleinberg's approach is slightly more general, as it does not assume that the two scores
are assigned to different types of entities.} using the web metaphor of ``mutual reinforcement'':
a page is authoritative if it is pointed by many good \emph{hubs}---pages which contain good list of authoritative pages---,
and a hub is good if it points to authoritative pages. This suggests an iterative process that computes
at the same time an authoritativeness score $\bm a_i$ and a ``hubbiness'' score $\bm h_i$ starting with
$\bm a_0=\mathbf 1$, and then applying the update rule
\begin{align*}
\bm h_{i+1} &= \bm a_i A^T	\\
\bm a_{i+1} &= \bm h_{i+1} A. 
\end{align*}
This process converges to the left dominant eigenvector of the matrix $A^TA$, which gives the final authoritativeness score,
and that we label with ``HITS'' throughout the paper.\footnote{As discussed
in~\cite{FLMRWARHPSEUEI}, the dominant eigenvector may not be unique;
equivalently, the limit of the recursive definition given above may depend on
the way the authority and hub scores are initialized. Here, we consider the result of the iterative process starting with $\bm a_0=\mathbf 1$.}

Inverting the process, and considering the left dominant eigenvector of the matrix $AA^T$, gives the final hubbiness score.
The two vectors are the left and right \emph{singular vectors} associated with the largest \emph{singular value}
in the singular-value decomposition of $A$. 
Note also that hubbiness is the positive version of authoritativeness.

\subsubsection{SALSA}
Finally, we consider SALSA, a measure introduced by Lempel and Moran~\cite{LeMS} always using the
metaphor of mutual reinforcement between authoritativeness and hubbiness, but $\ell_1$-normalizing the matrices $A$ and $A^T$.
We start with $\bm a_0=\mathbf 1$ and proceed with   
\begin{align*}
\bm h_{i+1} &= \bm a_i \overline{A^T}	\\
\bm a_{i+1} &= \bm h_{i+1} \bar A. 
\end{align*}
We remark that this normalization process is analogous to the one that brings us
from the dominant eigenvector to Seeley's index, or from Katz's index to PageRank.

Similarly to what happens with Seeley's index on
symmetric graphs, SALSA does not need such an iterative process to be computed.\footnote{This
property, which appears to be little known, is proved in Proposition 2 of the
original paper~\cite{LeMS}.} First, one computes the connected components of the symmetric graph
induced by the matrix $A^TA$; in this graph, $x$ and $y$ are
adjacent if $x$ and $y$ have some common predecessor in the
original graph.
Then, the SALSA score of a node is the ratio between its indegree and the sum of the
indegrees of nodes in the same component, multiplied by the ratio between
the component size and $n$.
Thus, contrarily to HITS, a single linear scan of the graph is sufficient to
compute SALSA, albeit the computation of the intersection graph requires time
proportional to $\sum_x d^+(x)^2$.

\subsection{Path-based measures}

Path-based measures exploit not only the existence of shortest paths but actually take
into examination all shortest paths (or all paths) coming into a node. We remark that
indegree can be considered a path-based measure, as it is the equivalent to the number of incoming
paths of length one.

\subsubsection{Betweenness}
\emph{Betweenness centrality} was introduced by Anthonisse~\cite{AntRG} for edges, 
and then rephrased by Freeman for nodes~\cite{FreSMCBB}. The idea is to measure
the probability that a random shortest path passes through a given node: if
$\sigma_{yz}$ is the number of shortest paths going from $y$ to $z$, 
and $\sigma_{yz}(x)$ is the number of such paths that pass through $x$, we define the
\emph{betweenness} of $x$ as 
\[
	\sum_{y,z \neq x,\sigma_{yz}\neq 0} \frac{\sigma_{yz}(x)}{\sigma_{yz}}.
\]
The intuition behind betweenness is that if a large fraction of shortest paths
passes through $x$, then $x$ is an important junction point of the network.
Indeed, removing nodes in betweenness order causes a very quick disruption
of the network~\cite{BRVRSWGNR}. 

\subsubsection{Spectral measures as path-based measures}
It is a general observation that all spectral measures can be interpreted as path-based measures
as they depend on taking the limit of some summations of powers of $A$, or on the limit of powers of $A$,
and in both cases we can express these algebraic operations in terms of suitable paths. 

For instance, the left dominant eigenvector of a nonnegative matrix can be
computed with the power method by taking the limit of $\mathbf 1 A^k/\|\mathbf 1 A^k\|$ for $k\to
\infty$. Since, however, $\mathbf 1A^k$ is a vector associating with each node the number of 
paths of length $k$ coming into the node, we can see that dominant eigenvector expresses
the relative growth of the number of paths coming into each node as their length increases.

Analogously, Seeley's index can be computed (modulo a normalization factor) by taking the limit of $\mathbf 1\bar A^k$ (in this
case, the $\ell_1$ norm cannot grow, so we do not need to renormalize at each iteration). The vector 
$\mathbf 1\bar A^k$ has the following combinatorial interpretation: it assigns to each $x$ the sums of the \emph{weights}
of the paths coming into $x$, where the weight of a path
$x_0$,~$x_1$, $\dots\,$, $x_t$ is
\begin{equation}
\label{eq:weight}
\prod_{i=0}^{t-1}\frac1{d^+(x_i)}.
\end{equation}

When we switch to attenuated versions of the previous indices (that is, Katz's index and PageRank), we
switch from limits to infinite summations and at the same time multiply the weight of paths of length $t$ by $\beta^t$ or $\alpha^t$.
Indeed, the Katz index of $x$ was originally defined as the summation over all $t$ of the number of paths of length $t$ coming into
$x$ multiplied by $\beta^t$, and PageRank is the summation over all paths coming into $x$ of the
weight~(\ref{eq:weight}) multiplied by $\alpha^t$.

The reader can easily work out similar definitions for
HITS and SALSA, which depend on a suitable definition of alternate ``back-and-forth path'' (see, e.g.,~\cite{BRRLAR})

\section{Axioms for Centrality}
\label{sec:axioms}

The comparative evaluation of centrality measures is a challenging, difficult,
arduous task, for many different reasons. The datasets that are classically used
in social sciences are very small (typically, some tens of nodes), and it is hard
to draw conclusions out of them. Nonetheless, some attempts in this direction were put forth
(e.g.,~\cite{SteRCME}); sometimes, the attitude was actually to provide evidence
that different measures highlight different kinds of centralities and are,
therefore, equally incomparably interesting~\cite{FriTFCM}. Whether the latter
consideration is the only sensible conclusion or not is debatable. While it is clear
that the notion of centrality, in its vagueness, can be interpreted differently
giving rise to many good but incompatible measures, we will provide evidence
that some measures tend to reward nodes that are in no way central.

If results obtained on small corpora may be misleading, a comparison on larger
corpora is much more difficult to deal with, due to the lack of ground truth and
to the unavailability of implementations of efficient algorithms to compute the
measures under consideration (at least in cases where efficient, possibly
approximate, algorithms do exist).
Among the few attempts that try a comparison on large networks we
cite~\cite{UCHQIEHPF} and~\cite{NZHT}, that nevertheless focus only on web
graphs and on a very limited number of centrality indices.

In this paper, we propose to understand (part of) the behavior of a centrality measure
using a set of axioms. While, of course, it is not sensible to prescribe a set
of axioms that \emph{define} what centrality should be (in the vein of Shannon's
definition of entropy~\cite{ShaMTC} or Altman and Tennenholtz axiomatic
definition of Seeley's index~\cite{AlTRSPA}\footnote{The authors claim to
formalize PageRank~\cite{PBMPCR}, but they do not consider the damping factor
(equivalently, they are setting $\alpha=1$), so they are actually formalizing
Seeley's venerable index~\cite{SeeNRI}.}), as different indices serve different
purposes, it is reasonable to set up some \emph{necessary} axioms that an index
should satisfy to behave predictably and follow our intuition.

The other interesting aspect of defining axioms is that, even if one does not believe
they are really so discriminative or necessary, they provide a very specific, formal, provable piece
of information about a centrality measure that is much more precise than folklore intuitions
like ``this centrality is really correlated to indegree'' or ``this centrality is really fooled by cliques''.
We believe that a theory of centrality should provide exactly this kind of compact,
meaningful, reusable information (in the sense that it can be used to prove other properties).
This is indeed what happens, for example, in topology,
where the information that a space is $T_0$, rather than $T_1$, is a compact way
to provide a lot of knowledge about the structure of the space.

Defining such axioms is a delicate matter. First of all, the semantics of the axioms must be very clear. Second,
the axioms must be evaluable in an exact way on the most common centrality measures.
Third, they should be formulated avoiding the trap of small, finite (counter)examples, on which many centrality measures collapse 
(e.g., using an asymptotic definition as suggested by Lempel and Moran for rank \emph{similarity} and \emph{stability}~\cite{LeMRSRS}).
We assume from the beginning that
the centrality measures under examination are invariant by isomorphism, that is, that they depend just on the 
structure of the graph, and not on particular labeling chosen for each node.

To meet these constraints, we propose to study the reaction of centrality
measures to \emph{change of size}, to \emph{(local) change of density} and 
to \emph{arc additions}. We expect that nodes
belonging to larger groups, when every other parameter is fixed, should be more
important, and that nodes with a denser neighborhood (i.e., having more
friends), when every other parameter is fixed, should also be more important. We also
expect that adding an arc should increase the importance of the target.

How can we actually determine if this happens in an exact way, and possibly in an asymptotic setting?
To do so, we need to try something entirely new---evaluating \emph{exactly} (i.e., in algebraic closed form) all measures of
interest on all nodes of some representative classes of networks.


A good approach to reduce the amount of computation is using
strongly connected \emph{vertex-transitive}\footnote{A graph is vertex-transitive if for every nodes $x$ and $y$ there is an automorphism
exchanging $x$ and $y$.} graphs as basic building blocks: these graphs 
exhibit a high degree of symmetry, which
should entail a simplification of our computations. Finally, since we want to compare
density, a natural choice is to pick the \emph{densest} strongly connected
vertex-transitive graph, the clique, and the \emph{sparsest} strongly connected,
the directed cycle.
Choosing two graphs at the extreme of the density spectrum should better 
highlight the reaction of centrality measures to density.
Moreover, $k$-cliques and directed $p$-cycles exist for every $k$ and
$p$ (this might not happen for more complicated structures, e.g., a cubic
graph).

\subsection{The size axiom}

Let us consider a graph made by a $k$-clique and a $p$-cycle (see the figure in Table~\ref{tab:cliquecycle}).\footnote{The graph is of course disconnected. It is a common
theme of this work that centrality measures should work also on graphs that are not strongly connected, for the
very simple reason that we meet this kind of graphs in the real world, the web being a common example.} 
Because of invariance by isomorphism, all nodes of the clique have equal score,
and all nodes of the cycle have equal score, too; but which nodes are
more important? Probably everybody would answer that if $p=k$ the elements on the clique are more important, and
indeed this axiom is so trivial that is satisfied by almost any measure of which we are aware, but we
are interested in assessing the sensitivity to \emph{size},\footnote{Lempel and Moran~\cite{LeMS} while designing their own centrality measure observe that
``the size of the community should be considered when evaluating the quality of the top pages in that community''.} and thus we state our first axiom as follows:

\begin{definition}[Size axiom]
\label{def:size}
Consider the graph $S_{k,p}$ made by a $k$-clique and a directed $p$-cycle (Table~\ref{tab:cliquecycle}).
A centrality measure satisfies the \emph{size axiom} if for every $k$ there is a $P_k$ such 
that for all $p\geq P_k$ in $S_{k,p}$ the centrality of a node of the $p$-cycle is strictly larger than
the centrality of a node of the $k$-clique, and if for every $p$ there is a $K_p$ such 
that for all $k\geq K_p$ in $S_{k,p}$ the centrality of a node of the $k$-clique is strictly larger than
the centrality of a node of the $p$-cycle.
\end{definition}

Intuitively, when $p=k$ we do expect nodes of the cycle to be less important
than nodes of the clique. The rationale behind the case $k\to\infty$ is
rather clear: the denser community is also getting larger, and thus its members
are expected to become ultimately even more important.

On the other hand, if the cycle becomes very large (more precisely, when its
size goes to infinity), its nodes are still part of a very large (albeit badly
connected) community, and we expect them to achieve ultimately greater
importance than the nodes of a fixed-size community, no matter how dense it can be.

Since one might devise some centrality measures that satisfy the size axiom for
$p$ and not for $k$, stating both properties in Definition~\ref{def:size} gives
us a finer granularity and avoids pathological cases.


\subsection{The density axiom}

Designing an axiom for density is a more delicate issue, since we must be able to
define an increase of density ``with all other parameters fixed'', including
size. Let us start ideally from a graph made
by a directed $k$-cycle and a directed $p$-cycle, and connect a node $x$ of the
$k$-cycle with a node $y$ of the $p$-cycle through a bidirectional
arc, the \emph{bridge}. If $k=p$, the vertices $x$ and $y$ are symmetric, and
thus must necessarily have the same score. Now, we increase the density of the
$k$-cycle as much as possible, turning it into a $k$-clique (see the figure in
Table~\ref{tab:cliquecyclexyyx}). Note that this change of density is local to
$x$, as the degree of $y$ has not changed. We are thus \emph{strictly increasing
the local density around $x$, leaving all other parameters fixed}, and in these
circumstances we expect that the score of $x$ increases.

\begin{definition}[Density axiom]
\label{def:density}
Consider the graph $D_{k,p}$ made by a $k$-clique and a $p$-cycle ($p,k\geq 3$) connected by a bidirectional bridge
$x \leftrightarrow y$, where $x$ is a node of the clique and $y$ is a node of the cycle (Table~\ref{tab:cliquecyclexyyx}).
A centrality measure satisfies the \emph{density axiom} if for $k=p$ the centrality
of $x$ is strictly larger than the centrality of $y$.
\end{definition}

Note that our axiom does not specify any constraint when $k\neq p$. While
studying the behavior of the graph $D_{k,p}$ of the previous definition when
$k\neq p$ shades some lights of the inner behavior of centrality measures, it
is essential, in an axiom asserting the sensitivity to density, that size is not
involved.

In our proofs for the density axiom, we actually let $k$ and $p$ be independent
parameters (even if the axiom is stated for $k=p$) to compute the \emph{watershed},
that is, the value of $k$ (expressed as a function of $p$) at which the 
centrality of $x$ becomes larger than the centrality of $y$ (if any).
The watershed can give some insight as to how badly a measure can miss
satisfying the density axiom.

\subsection{The score-monotonicity axiom}
\label{sec:mono}

Finally, we propose an axiom that specifies strictly monotonic behavior upon the addition of an arc:
 
\begin{definition}[Score--Monotonicity Axiom]
\label{def:density1}
A centrality measure satisfies the \emph{score-monotonicity axiom} if
for every graph $G$ and every pair of nodes $x$, $y$ such that $x\not\to y$,
when we add $x\to y$ to $G$ the centrality of $y$ increases.
\end{definition}

In some sense, this axiom is trivial: it is satisfied by essentially all centrality
measures we consider \emph{on strongly connected graphs}. Thus, it is
an excellent test to verify that a measure is able to handle correctly partially disconnected graphs.

The reader might be tempted to define a \emph{weak} monotonicity
axiom which just require the \emph{rank} of $y$ to be nondecreasing. However,
the centrality measure associating a constant value to every node of every
network would satisfy such an axiom, which makes it not very interesting for our
goals if used in isolation.

Nonetheless, a form of rank monotonicity would be an essential counterpoint to score monotonicity, as it would detect
pathological centralities in which the addition of an arc towards $y$ does increase the score of
$y$, but also increase in a counterintuitive way the score of other nodes,
resulting in a rank decrease of $y$.
We leave the statement and the study of such an axiom for future work.

\subsection{Previous attempts at axiomatization}
\label{sec:histax}

The idea of formalizing through axioms the behavior of centrality indices is not new:
Sabidussi, in his much-quoted paper on centrality~\cite{SabCIG},\footnote{This paper contributed so 
much to the popularization of Bavelas's closeness centrality that the latter is often called ``Sabidussi's centrality''
in the literature.} describes in a section named ``Axioms for centrality''
a set of axioms that should be satisfied by a sensible centrality on an undirected graph. Indeed, one of his conditions is
identical to our score-monotonicity axiom (see Section~\ref{sec:mono}).

A few years later, Nieminen~\cite{NieCDG} attempted a similar formalization for directed graphs. His axioms
are only meaningful for geometric centralities as they are stated in term of the number of nodes at a given distance.
Consider the \emph{negative neighborhood function of $x$ in a graph $G$} defined by \[N^-_G(x,t)=\bigl|\{y\,\mid d(y,x)\leq t\,\}\bigr|.\] Nieminen
essentially states that if the neighborhood function of $x$ dominates strictly that of $y$ (i.e., $N^-_G(x,t)\geq N^-_G(y,t)$ for all $t$ and $N^-_G(x,t)> N^-_G(y,t)$ for at least one $t$)
then the centrality of $x$ must be strictly larger than that of $y$. Since the new neighborhood function of $y$ after
the addition of an arc towards $y$ always dominates strictly the old one, Nieminen's axiom implies score monotonicity. 
He notes, indeed, that Bavelas's centrality does not satisfy some of his axioms.

The work of Chien, Dwork, Kumar, Simon and Sivakumar~\cite{CDKLEAA}, albeit
targeted exclusively at PageRank, contains a definition of score monotonicity
equivalent to ours (and Sabidussi's) and a definition of rank monotonicity that
we think captures the essence of the problem:
when we add an arc towards $y$, nodes with a score smaller than $y$ must continue to have a
score smaller that of $y$, while nodes with a score equal to $y$ must get a score
that is smaller than or equal to that of $y$. 

Sabidussi's paper tries to capture rank monotonicity, too, but with much weaker
requirements: he requires that if $y$ has maximum score in the network, then it
should have maximum score also after the addition of an arc towards $y$.

Very recently, Brandes, Kosub and Nick~\cite{BKNWMZ} proposed to call
\emph{radial}\footnote{There are actually two notion of radiality, which
correspond to our notion of ``positive'' and ``negative'' centralities.} a
centrality in which the addition of an arc $x\to y$ does not decrease the rank
of $y$, in the sense that when we add an arc towards $y$, nodes with a score
smaller than or equal to $y$ continue to have this property. This is actually a weaker condition than the one used
in~\cite{CDKLEAA} to state rank monotonicity; in fact, it makes a form of non-monotonicity of ranks possible,
as nodes with a score smaller than $y$ might end up having a score equal to $y$ when we add
an arc towards $y$. The authors of~\cite{BKNWMZ} propose also a different axiom for undirected networks.

\section{Proofs and Counterexamples}

We have finally reached the core of this paper: since we are considering
eleven centralities and three axioms, we have to verify 33 statements. For the
size and density axioms, we compute in closed form the values of all measures, from which
we can derive the desired results,
whereas for the score-monotonicity axiom we provide direct proofs or counterexamples.

We remark that in all our tables we use the proportionality symbol $\propto$ to mark
values that have been rescaled by a common factor to make them more readable.

\subsection{Size}

Table~\ref{tab:cliquecycle} provides scores for the graph $S_{p,k}$, from which
we can check whether the size axiom is satisfied. The scores are immediately
computable from the basic definitions; as we noticed, $S_{k,p}$ is highly symmetrical, and so there
are only two actual scores---the score of a node of the clique and the score of a node of the cycle.
In the case of some spectral centrality measures
there are several possible solutions, in which case we use the one returned
by the power method starting from the uniform vector.

\begin{table}
\renewcommand{\arraystretch}{2.5}
\centering
\includegraphics{fig/harmonic-2.mps}
\begin{tabular}{l|c|c}
Centrality & $k$-clique & $p$-cycle \\
\hline
Degree & 
	$k-1$ & 
	1 \\
Harmonic & 
	$k-1$ & 
	$H_{p-1}$ \\
Closeness& 
	$\displaystyle\frac1{k-1}$ & 
	$\displaystyle\frac2{p(p-1)}$ \\
Lin & 
	$\displaystyle\frac{k^2}{k-1}$  &
	$\displaystyle\frac{2p}{p-1}$ \\ 
Betweenness & 
	0 &
	$\displaystyle\frac{(p-1)(p-2)}2$\\ 
Dominant $\propto$ & $1$ & $0$\\
Seeley $\propto$ & 
	1 & 
	1 \\
Katz & $\displaystyle\frac1{1-(k-1)\beta}$ & $\displaystyle\frac1{1-\beta}$\\
PageRank $\propto$ & $1$ & $1$\\
HITS $\propto$ & $1$ & $0$\\
SALSA $\propto$ & $1$ & $1$\\
\end{tabular}
\caption{\label{tab:cliquecycle}Centrality scores for the graph $S_{k,p}$. $H_i$ denotes the $i$-th harmonic number. The parameter $\beta$ is Katz's attenuation factor.}
\end{table}

\subsection{Density}

Table~\ref{tab:cliquecyclexyyx} provides scores for the graph $D_{p,k}$. Since the graph is strongly
connected, there is no uniqueness issue. 
While the computation of geometric and path-based centrality measures
is a tedious but rather straightforward exercise (it is just a matter
of finite summations), spectral indices require some more care. In the rest of this section, we shall first detail
the steps leading to the explicit formulae of Table~\ref{tab:cliquecyclexyyx}. Then,
we will prove that the density axiom holds. 

\subsubsection{Explicit formulae for spectral indices}

We will write the parametric equations expressing the matrix computation that
defines each spectral index and solve them. As noted before, even if the axiom requires $k=p$ we prefer to perform
the computation with two independent parameters $k$ and $p$ to compute the watershed.

In all cases, we can always use the bounds imposed by symmetry to write down just a small number of variables: $c$ for the centrality of an element of the clique,
$\ell$ for the clique bridge (``left''), $r$ for the cycle bridge (``right''), and some function $t(d)$ of the distance from
the cycle bridge for the nodes of the cycle (with $0\leq d<p$), with the condition $t(0)=r$.

\fakepar{The left dominant eigenvector} In this case, the equations are given by the standard eigenvalue problem of the adjacency matrix:
\begin{align*}
\lambda \ell &= r+(k-1)c\\
\lambda c &= \ell+(k-2)c\\
\lambda r &= \ell+\frac r{\lambda^{p-1}},
\end{align*}
subject to the condition that we choose $\lambda$ real and positive with maximum absolute value.
Note that in the case of the last equation we ``unrolled'' the equations about the elements of 
the cycle, $\lambda t(d+1)=t(d)$. Solving the system and choosing $c=1/(\lambda-k+1)$ gives the 
solutions found in Table~\ref{tab:cliquecyclexyyx}. 

Since for nonnegative matrices the dominant eigenvalue is monotone in the matrix entries, $\lambda\geq k-1$, because the 
$k$-clique has dominant eigenvalue equal to $k-1$. On the
other hand, $\lambda\leq k$ by the row-sum bound. As the eigenvalue equations have no solution for $\lambda=k-1$,  we
conclude that $k-1<\lambda \leq k$.

\fakepar{Katz's index} In this case, the equations can be obtained by the standard technique of ``taking one summand out'',
that is, writing
\[\bm k= \mathbf 1\sum_{i=0}^\infty \beta^iA^i = \mathbf 1+ \mathbf 1\sum_{i=1}^\infty \beta^iA^i = \mathbf 1+ \biggl(\mathbf 1\sum_{i=0}^\infty \beta^iA^i\biggr)\beta A = \mathbf 1+ \bm k\beta A.\]
The equations are then
\begin{align*}
\ell &=1+\beta r+\beta (k-1)c\\
c&=1+\beta\ell+\beta( k- 2 )c \\
r&=1+\beta \ell + \beta\biggl( {\frac {1-{\beta}^{p-1}}{1-\beta}}+{\beta}^{p-1}r \biggr),
\end{align*}
where again we ``unrolled'' the equations about the elements of 
the cycle, as we would have just $t(d+1)=1+\beta t(d)$, so 
\[
t(d) = \frac{1-\beta^d}{1-\beta}+\beta^d r.
\]
The explicit values of the solutions are quite ugly, so we present them in Table~\ref{tab:cliquecyclexyyx} as a function of the centrality of the clique bridge $\ell$.

\fakepar{PageRank} To simplify the computation, we use $\mathbf 1$, rather than $\mathbf 1/(k+p)$, as preference vector 
(by linearity, the result obtained just needs to be rescaled).
We use the same technique employed in the computation of Katz's index, leading to 
\begin{align*}
\ell &= 1-\alpha+\frac12\alpha r+\alpha c\\
c &= 1-\alpha+\frac\alpha k \ell+\alpha\frac{k-2}{k-1}c\\
r &= 1-\alpha+\frac\alpha k\ell+\alpha\biggl(1-\alpha^{p-1}+\frac12 \alpha^{p-1}r\biggr),
\end{align*}
noting once again that unrolling the equation of the cycle $t(1)=1-\alpha+\alpha r/2$ and $t(d+1)=1-\alpha + \alpha t(d)$ for $d>1$ we get
\[
t(d)=1-\alpha^d+\frac12\alpha^d r.
\]
The explicit values for PageRank are even uglier than those of Katz's index, so again we present them in Table~\ref{tab:cliquecyclexyyx} as a
function of the centrality of the clique bridge $\ell$.

\fakepar{Seeley's index} This is a freebie as we can just compute PageRank's limit when $\alpha\to1$.

\fakepar{HITS} In this case, we write down the eigenvalue problem for $A^TA$. Writing $t$ for $t(1)$, we have
\begin{align*}
\mu c & = (k-1)c+(k-2)^2c+(k-2)\ell+r\\
\mu\ell &=  k\ell+(k-1)(k-2)c+t\\
\mu r &= 2r+(k-1)c\\
\mu t &= t+\ell.
\end{align*}
By normalizing the result so that $c=\mu^2 -\mu(k+1)+k-1$, 
we obtain the complex but somewhat readable
values shown in Table~\ref{tab:cliquecyclexyyx}. Note that $p$ has no role in the solution, because $A^TA$ can be decomposed into two
independent blocks, one of which is an identity matrix corresponding to 
all elements of the cycle except
for the first two.

\fakepar{SALSA} 
It is easy to check that the components of the intersection graph of predecessors are given by the clique together with the
cycle bridge and its successor, and then by one component for each node of the cycle. The computation of the scores
is then trivial using the non-iterative rules.

\begin{sidewaystable*}
\centering
\renewcommand{\arraystretch}{2.5}
\includegraphics{fig/harmonic-5.mps}
\begin{tabular}{l|c|c|c|c|c|c}
Centrality & Clique & Clique bridge & Cycle bridge & Cycle ($d>0$ from the
bridge)&  Watershed  
\\
\hline
Degree &
	$k-1$ &
	$k$ &
	$2$ &
	$1$ &
 	---
\\
Harmonic & 
	$k-2+H_{p+1}$ & $k - 1 + H_p$ &
	$\displaystyle1+\frac{k-1}{2}+ H_{p-1}$  &
	$\displaystyle\frac1{d+1} +\frac{k-1}{d+2}+ H_{p-1}$  &
	---  
\\
Closeness & 
	$\displaystyle\frac1{k-1+2p+p(p-1)/2}$& 
	$\displaystyle\frac1{k-1+p+p(p-1)/2}$& 
	$\displaystyle\frac1{2k-1+p(p-1)/2}$&
	$\displaystyle\frac1{k(d+2)-1+p(p-1)/2}$&
	$k\leq p$
\\
Betweenness & 
	0 & $2p(k-1)$ & 
	$\displaystyle 2k(p-1)+\frac{(p-1)(p-2)}2$ &
	$\displaystyle 2k(p-2)+\frac{(p-1)(p-2)}2$ &
	$\displaystyle k\leq \frac{p^2+p+2}{4}$ 
 \\
Dominant & 
 $\displaystyle\frac1{\lambda-k+1}$ &$\displaystyle1+\frac1{\lambda-k+1}$ &$1+\lambda$&$\displaystyle\frac{1+\lambda}{\lambda^d}$&
---
\\
Seeley $\propto$ & 
	$\displaystyle k-1$&
	$\displaystyle k$&
	$\displaystyle 2$&
	$\displaystyle 1$&
	---
\\
Katz $\propto$  & 
 	$\displaystyle\frac{1+\beta \ell}{1-\beta(k-2)}$&
 	$\ell$&
 	$\displaystyle\frac1{1-\beta}+\frac\beta{1-\beta^p}\ell $&
 	$\displaystyle\frac1{1-\beta}+\frac{\beta^{d+1}}{1-\beta^p}\ell$&
 	---
\\ 
PageRank $\propto$ & 
	$\displaystyle\frac{(k-1)(k-\alpha k+\alpha \ell)}{k(k-1-\alpha(k-2))}$&
	$\ell$&
	$\displaystyle 2+2\frac{\alpha \ell-k}{k(2-\alpha^p)}$&
	$\displaystyle 1+\alpha^d\frac{\alpha \ell-k}{k(2-\alpha^p)}$&
	---
\\
HITS $\propto$&
 $\mu^2-\mu(k+1)+k-1$ &
  $(k-1)(k-2)(\mu-1)$ &
\begin{minipage}{5cm}\begin{multline*}\mu^3
-(k^2-2k+4)\mu^2\\+(3k^2-7k+6)\mu -(k-1)^2\end{multline*}\end{minipage}&
$[d=1](k-1)(k-2)$ & ---
\\
SALSA $\propto$ & $\displaystyle(k-1)(k+2)$ &
$\displaystyle k(k+2)$ & $\displaystyle2(k+2)$& $\displaystyle k+2 +[d\neq1](k^2-2k+2)$ &
---
\\
\end{tabular}
\caption{\label{tab:cliquecyclexyyx}Centrality scores for the graph $D_{k,p}$.
The parameter $\beta$ is Katz's attenuation factor, $\alpha$ is PageRank's damping factor, $\lambda$
is the dominant eigenvalue of the adjacency matrix $A$ and $\mu$
is the dominant eigenvalue of the matrix $A^TA$.  Lin's centrality is omitted because it is
proportional to closeness (the graph being strongly connected).}
\end{sidewaystable*}

\subsubsection{Proofs}

Armed with our explicit formulation of spectral scores, we have now to
prove whether the density axiom holds, that is, whether $\ell>r$ when $k=p$.
Note that in Table~\ref{tab:cliquecyclexyyx} we report no watershed for all
spectral centrality measures, which means even more: $\ell>r$
even when $k\neq p$, provided that $k,p\geq 3$. The proofs in this section
cover this stronger statement.

\begin{theorem}
\label{teo:HITSdensity}
HITS satisfies the density axiom.
\end{theorem}
\begin{proof}
As we have seen, we can normalize the solution to the HITS equations so that
\begin{align*}
\ell &= (k-1)(k-2)(\mu-1) \\
r &= \mu^3 -(k^2-2k+4)\mu^2+(3k^2-7k+6)\mu -(k-1)^2
\end{align*}
Moreover, the characteristic polynomial can be computed explicitly from the set of equations and 
by observing that the vectors $\bm\chi_{k+i}$ and $\bm\chi_0-\bm\chi_i$, $0< i< p$, are linearly independent
eigenvectors for the eigenvalue $1$: 
\[
p(\mu) = \bigl({\mu}^{4}- ( {k}^{2}-2k+6 ) {\mu}^{3}+ ( 5 {k}^{2}-12k+15 ) {\mu}^{2} -( 6{k}^{2}-16k+14 ) \mu+{k}^{2}-2k+1\bigr)(\mu-1)^{k+p-4}.
\]
The largest eigenvalue $\mu_0$ satisfies the inequality $(k-1)^2\leq\mu_0\leq k^2-2k+5/4$ for every $k\geq 9$ as shown below 
(the statement of the theorem can be verified in the remaining cases by explicit computation, as it does not depend on $p$).
Using the stated upper and lower bounds on $\mu_0$, we can say that 
\begin{eqnarray*}
\lefteqn{\ell-r=(k-1)(k-2)(\mu-1) - ( \mu^3 -(k^2-2k+4)\mu^2+(3k^2-7k+6)\mu -(k-1)^2 )  }\\
 &=& -{\mu}^{3}+ ( {k}^{2}-2k+4 ) {\mu}^{2}- ( 2
{k}^{2}-4k+4 ) \mu+k-1\\
&\geq&- \biggl( {k}^{2}-2k+\frac54 \biggr) ^{3}+ ( {k}^{2}-2k+4 ) 
 ( k-1 ) ^{4}- ( 2{k}^{2}-4k+4 )  \biggl( {k
}^{2}-2k+\frac54 \biggr) +k-1\\
&=&\frac14k^4-k^3-\frac{19}{16}k^2+\frac{43}{8}k-\frac{253}{64},
\end{eqnarray*}
which is positive for $k>4$. 

We are left to prove the bounds on $\mu_0$. The lower bound can be easily
obtained by monotonicity of the dominant eigenvalue in the matrix entries,
because the dominant eigenvalue of a $k$-clique is $k-1$.
For the upper bound, first we observe that $\mu_0$ can be computed
explicitly (as it is the solution of a quartic equation) and using its
expression in closed form it is possible to show that
$\lim_{k\to\infty}\mu_0 - (k-1)^2=0$. This guarantees that the bound
$\mu_0\leq k^2-2k+5/4$ is true ultimately. To obtain an explicit value of
$k$ after which the bound holds true, observe that $k^2-2k+5/4=\mu_0$
implies $q(k)=0$, where $q(k)=p(k^2-2k+5/4)$. Computing the Sturm
sequence associated with $q(k)$ one can prove that $q(k)$ has no zeroes for $k\geq 9$,
hence our lower bound on $k$.
\end{proof}

\begin{theorem}
The dominant eigenvector satisfies the density axiom.
\end{theorem}
\begin{proof}
From the proof of Theorem~\ref{teo:HITSdensity} we know that $\lambda^2\leq k^2-2k+5/4$ for every $k\geq 9$, because
$\sqrt\mu$ is the spectral norm of $A$ and thus dominates its spectral radius $\lambda$, that is, $\lambda^2\leq \mu$. 
We conclude that
\begin{multline*}
\ell - r = 1 + \frac 1{\lambda-k+1} - ( 1 + \lambda ) =\frac{-\lambda^2+(k-1)\lambda + 1}{\lambda -k + 1}\\
>\frac{-(k^2-2k+5/4)+(k-1)^2 + 1}{\lambda -k + 1} = \frac3{4(\lambda -k + 1)}>0. 
\end{multline*}
The remaining cases ($k<9$) are verifiable by explicit computation when $p\leq k$. To prove that there is no watershed,
we first note that $\ell - r$ is decreasing in $\lambda$. Then, one can use the characteristic equation defining $\lambda$ and 
implicit differentiation to write down the values of the derivative of $\lambda$ as a function of $p$ for each $k<9$.
It is then easy to verify that in the range $p>k$, $k-1<\lambda\leq k$ the derivative is always negative, that is, $\lambda$ is decreasing in $p$,
which completes the proof.
\end{proof}

\begin{theorem}
\label{teo:katzdensity}
Katz's index satisfies the density axiom when $\beta\in(0\..1/\lambda)$.
\end{theorem}
\begin{proof}
Recall that the equations for Katz's index are
\begin{align*}
\ell &=1+\beta r+\beta (k-1)c\\
c&=1+\beta\ell+\beta( k- 2 )c \\
r&=1+\beta \ell + \beta\biggl( {\frac {1-{\beta}^{p-1}}{1-\beta}}+{\beta}^{p-1}r \biggr).
\end{align*}
First, we remark that as $\beta\to 1/\lambda$ Katz's index tends to the dominant eigenvector, 
so $\ell>r$ for $\beta$ close enough to $1/\lambda$. Thus, by continuity, we just need to show that $\ell=r$ never happens in the range of our parameters. If we
solve the equations above for $c$, $\ell$ and $r$ and impose $\ell=r$, we obtain
\[
p=\frac{\displaystyle\ln\frac{\beta^2+k-2}{k-1}}{\ln \beta}.
\]
Now observe that 
\[\beta\leq\frac{\beta^2+k-2}{k-1}\]
is always true for $\beta\leq 1$ and $k\geq 3$. This implies that under the same conditions $p\leq 1$, which concludes the proof.
\end{proof}

\begin{theorem}
PageRank with constant preference vector satisfies the density axiom when $\alpha\in (0\.. 1)$.
\end{theorem}
\begin{proof}
The proof is similar to that of Theorem~\ref{teo:katzdensity}. Recall that the equations for PageRank are
\begin{align*}
\ell&=1-\alpha+\frac12\alpha r+\alpha c\\
c&=1-\alpha+\frac{\alpha}k\ell+
\frac{\alpha(k-2)}{k-1}c \\
r&=1-\alpha+\frac{\alpha}k\ell+\alpha\biggl( 1-
\alpha^{p-1}+\frac12\alpha^{p-1}r \biggr).
\end{align*}
First, we remark that as $\alpha\to 1$ PageRank tends to Seeley's index, so  
$\ell>r$ for $\alpha$ close enough to $1$. By continuity, we thus just need to show that $\ell=r$ never happens in our range of parameters. If we
solve the equations above for $c$, $\ell$ and $r$ and impose $\ell=r$, we obtain
\[
p=1+\frac{\displaystyle\ln\biggl(-\frac{2\alpha^2-(k^2-4k+6)\alpha+k^2-3k+2}{(k^2-3k+2)\alpha^2-(2k^2-3k)\alpha+k^2-k}\biggr)}{\ln \alpha}.
\]
Now observe that $2\alpha^2-(k^2-4k+6)\alpha+k^2-3k+2\geq 0$ for $k\geq 3$. Thus, a solution for $p$ exists only when
the denominator is negative. However, in that region 
\[-\frac{2\alpha^2-(k^2-4k+6)\alpha+k^2-3k+2}{(k^2-3k+2)\alpha^2-(2k^2-3k)\alpha+k^2-k}\geq 1.\]
This implies that under the same conditions $p\leq 1$, which concludes the proof.
\end{proof}

\subsection{Score Monotonicity}

In this section, we briefly discuss the only nontrivial cases. 

\fakepar{Harmonic} If we add an arc $x\to y$ the harmonic centrality of $y$ can only increase, because
  this addition can only decrease distances (possibly even turning some of them from infinite to finite), so
  it will increase their reciprocals (strictly increasing the one from $x$). 
  
\fakepar{Closeness} If we consider a one-arc graph $z\to y$ and add an arc $x\to y$, the closeness
  of $y$ decreases from $1$ to $1/2$.
  
\fakepar{Lin} Consider the graph in Figure~\ref{fig:counter}: the Lin centrality of $y$ is $(k+1)^2/k$.
  After adding an arc $x \to y$, the centrality becomes $(k+5)^2/(k+9)$, which is smaller
  than the previous value when $k>3$.
\begin{figure}
\centering
\includegraphics{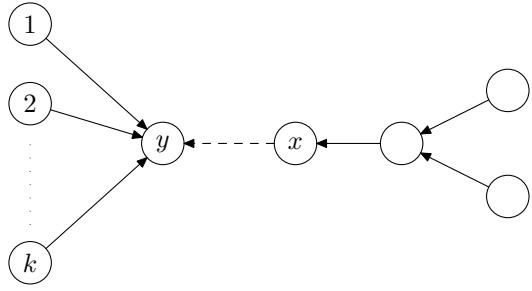}
\caption{\label{fig:counter}A counterexample showing that Lin's index fails to satisfy the score-monotonicity axiom.}
\end{figure}
  
\fakepar{Betweenness} If we consider a graph made of two isolated nodes $x$ and $y$, the addition of
  the arc $x\to y$ leaves the betweenness of $x$ and $y$ unchanged. 

\fakepar{Katz} The score of $y$ after adding $x \to y$ can only increase because 
  the set of paths coming into $y$ now contains new elements.\footnote{It should be noted, however,
  that this is true only for the values of the parameter $\beta$ that still make sense after the
  addition.} If the constant vector $\mathbf 1$ is replaced by
  a preference vector $\bm v$ in the definition, it is necessary
  that $x$ have nonzero score before the addition for score monotonicity to hold. 

\fakepar{Dominant eigenvector, Seeley's index, HITS} If we consider a clique and two isolated nodes $x$, $y$, the score given by
  the dominant eigenvector, Seeley's index and HITS to $x$ and $y$ is zero, and it remains 
  unchanged when we add the arc $x\to y$.

\fakepar{SALSA} Consider the graph in Figure~\ref{fig:countersalsa}: the indegree of $y$ is $1$, and its component in the
intersection graph of predecessors is trivial, so its SALSA centrality is $(1/1) \cdot (1/6)=1/6$.
After adding an arc $x \to y$, the  indegree of $y$ becomes $2$, but now its component is $\{\,y,z\,\}$; so the 
sum of indegrees within the component is $2+3=5$, hence the centrality of $y$ becomes $(2/5)\cdot (2/6)=2/15<1/6$.
\begin{figure}
\centering
\includegraphics{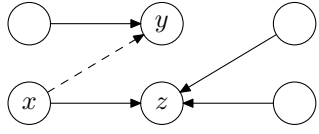}
\caption{\label{fig:countersalsa}A counterexample showing that SALSA fails to satisfy the score-monotonicity axiom.}
\end{figure}

\fakepar{PageRank}
Score monotonicity of PageRank was proved by Chien, Dwork, Kumar, Simon and
Sivakumar~\cite{CDKLEAA}. Their proof works for a generic \emph{regular} Markov chain:
in the case of PageRank this condition is true, for instance, if the preference vector is
strictly positive or if the graph is strongly connected. Score monotonicity under the same hypotheses is also a
consequence of the analysis made by Avrachenkov and Litvak~\cite{AvLENLGP} of the behavior of
PageRank when multiple new links are added.

Their result can be extended to a much more general setting. Suppose
that we are adding the arc $x \to y$, with the proviso that the PageRank of $x$ before adding the arc was strictly positive.
We will show that under this condition the score of $y$ will increase for \emph{arbitrary graphs and preference vectors}.
The same argument shows also that the if the score of $x$ is zero, the score of $y$ does not change.

For this proof, we define PageRank as $\bm v\bigl(1-\alpha\bar A\bigr)^{-1}$ (i.e., without the normalizing factor $1-\alpha$),
so to simplify our calculations. By linearity, the result for the standard definition follows immediately.

Consider two nodes $x$ and $y$ of a graph $G$ such that there is no arc from $x$ to $y$, and let $d$ be the outdegree of $x$.
Given the normalized matrix $\bar A$ of $G$, and the normalized matrix $\bar A'$ of the graph $G'$ obtained by adding to $G$ the arc
$x\to y$, we have
\[\bar A -\bar A' = \bm \chi_x^T\bm \delta,\] 
where $\bm\delta$ is the
difference between the rows corresponding to $x$ in $\bar A$ and $\bar A'$, which contains $1/d(d+1)$ in the positions corresponding to the successors of $x$ in $G$, and $-1/(d+1)$ in the position corresponding to $y$ (note that if $d=0$, we have
just the latter entry). 

We now use the Sherman--Morrison formula to write down the inverse of $1-\alpha \bar A'$ as a function of $1-\alpha\bar A$. More precisely,
\begin{multline*}
\bigl(1-\alpha\bar A'\bigr)^{-1} = \Bigl(1-\alpha\bigl(\bar A -\bm \chi_x^T\bm \delta \bigr)\Bigr)^{-1} = \bigl(1-\alpha\bar A +\alpha\bm \chi_x^T\bm \delta\bigr)^{-1}
\\=\bigl(1-\alpha\bar A\bigr)^{-1} -
\frac{\bigl(1-\alpha\bar A\bigr)^{-1}\alpha\bm \chi_x^T\bm \delta \bigl(1-\alpha\bar A\bigr)^{-1}}{1 + \alpha\bm\delta \bigl(1-\alpha\bar A\bigr)^{-1} \bm \chi_x^T}.   
\end{multline*}

We now multiply by the preference vector $\bm v$, obtaining the explicit PageRank correction:
\begin{multline*}
\bm v\bigl(1-\alpha\bar A'\bigr)^{-1} = \bm v\bigl(1-\alpha\bar A\bigr)^{-1} - \bm v
\frac{\bigl(1-\alpha\bar A\bigr)^{-1}\alpha\bm \chi_x^T\bm \delta \bigl(1-\alpha\bar A\bigr)^{-1}}{1 + \alpha\bm\delta \bigl(1-\alpha\bar A\bigr)^{-1} \bm \chi_x^T}\\=
\bm r -  \frac{ \alpha\bm r\bm \chi_x^T\bm \delta \bigl(1-\alpha\bar A\bigr)^{-1}}{1 + \alpha\bm\delta \bigl(1-\alpha\bar A\bigr)^{-1} \bm \chi_x^T}=
\bm r -  \frac{ \alpha r_x\bm \delta \bigl(1-\alpha\bar A\bigr)^{-1}}{1 + \alpha\bm\delta \bigl(1-\alpha\bar A\bigr)^{-1} \bm \chi_x^T}.	
\end{multline*}
Now remember that $r_x>0$, and 
note that $\bigl(1-\alpha\bar A\bigr)^{-1} \bm \chi_x^T$ is the vector of positive contributions to the PageRank of $x$, modulo the normalization
factor $1-\alpha$. As such, it is made of positive
values adding up to at most $1/(1-\alpha)$. When the vector is multiplied by $\bm\delta$, in the worst case ($d=0$) we obtain $1/(1-\alpha)$, so
given the conditions on $\alpha$ it is easy to see that the denominator is
positive. This implies that we can gather all constants in a single
positive constant $c$ and just write
\[
\bm v\bigl(1-\alpha\bar A'\bigr)^{-1} = \bigl(\bm v - c\bm\delta\bigr)\bigl(1-\alpha\bar A\bigr)^{-1}.
\]    
The above equation rewrites the rank-one correction due to the addition of the arc $x\to y$ as a formal correction of the preference vector. We
are interested in the difference 
\[\bigl(\bm v - c\bm\delta\bigr)\bigl(1-\alpha\bar A\bigr)^{-1} -\bm v\bigl(1-\alpha\bar A\bigr)^{-1} = - c\bm\delta\bigl(1-\alpha\bar A\bigr)^{-1},\]
as we can conclude our proof by just showing that its $y$-th coordinate is strictly positive.
 
We now note that being $\bigl(1-\alpha\bar A\bigr)$ strictly diagonally dominant, the (nonnegative) inverse $B=\bigl(1-\alpha\bar A\bigr)^{-1}$ has the property
that the entries $b_{ii}$ on the diagonal are strictly larger than off-diagonal 
entries $b_{ki}$ on the same column~\cite[Remark 3.3]{MNSIMIGUM}, and in particular they are nonzero. Thus, if $d=0$
\[\bigl[-c\bm\delta\bigl(1-\alpha\bar A\bigr)^{-1}\bigr]_y  = \frac c{d+1} b_{yy} > 0,\]
and if $d\neq0$
\[\bigl[-c\bm\delta\bigl(1-\alpha\bar A\bigr)^{-1}\bigr]_y  = \frac c{d+1} b_{yy} - \sum_{x\to z}\frac c{d(d+1)} b_{zy}>\frac c{d+1} b_{yy} - \sum_{x\to z}\frac c{d(d+1)} b_{yy}= 0.\]

\medskip
We remark that the above discussion applies to PageRank as defined by~(\ref{eq:prinv}):
if the scores are forced to be $\ell_1$-normalized, the score-monotonicity axiom may fail to hold even 
under the assumption that the score of $x$ is positive. If we take, for example, the graph with adjacency matrix
\[
	A=\left(\begin{array}{cc}0 & 0\\ 1 & 0\end{array}\right)
\]
and the preference vector $\bm v = (0,1)$, we have $\bm p=(\alpha(1-\alpha),1-\alpha)$.
Adding one arc from the first node to the second one yields  $\bm p=(\alpha/(1+\alpha),1/(1+\alpha))$: the score  
of the second node increases ($1/(1+\alpha)$ is always larger than $1-\alpha$) as stated in the theorem. However, the \emph{normalized} PageRank score
of the second node does not change (it is equal, in both cases, to $1/(1+\alpha)$).

\section{Roundup}

All our results are summarized in Table~\ref{tab:axiomatic}, where we distilled
them into simple yes/no answers to the question: does a given centrality measure satisfy the axioms?

It was surprising for us to discover that \emph{only harmonic centrality satisfies all axioms}.\footnote{It is interesting
to note that it is actually the only centrality satisfying the size axiom---in fact, one needs a cycle
of $\approx e^k$ nodes to beat a $k$-clique.}
All spectral centrality measures are sensitive to density.
Row-normalized spectral centrality measures (Seeley's index, PageRank and SALSA) are insensitive to size, whereas
the remaining ones are only sensitive to the increase of $k$ (or $p$ in the
case of betweenness). All non-attenuated spectral measures are also non-monotone.
Both Lin's and closeness centrality fail density tests.\footnote{We note that since $D_{k,p}$ is strongly connected, closeness and Lin's 
centrality differ just by a multiplicative constant.}
Closeness has, indeed, the worst possible behavior, failing to satisfy all our axioms. While this result might seem
counterintuitive, it is actually a consequence of the known tendency 
of very far nodes to dominate the score, hiding the contribution of closer nodes, whose presence is more correlated
to local density.

All centralities satisfying the density axiom have no watershed: the axiom is satisfied for all $p,k\geq 3$.
The watershed for closeness (and Lin's index) is $k\leq p$, meaning that they just miss it, 
whereas the watershed for betweenness is a quite
pathological condition ($k\leq (p^2+p+2)/4$): one needs a clique whose size is \emph{quadratic} in the size of the cycle before
the node of the clique on the bridge becomes more important than the one on the cycle (compare this with closeness, where $k=p+1$ is sufficient).

We remark that our results on geometric indices do not change if we replace the directed cycle with a symmetric (i.e., undirected) cycle, with
the additional condition that $k>3$. It 
is possible that the same is true also of spectral centralities, but the geometry of the paths of the 
undirected cycle makes it extremely difficult to carry on the analogous computations in that case. 

\begin{table}
\centering
\begin{tabular}{l|c|c|c}
Centrality & Size & Density & Score monotonicity \\
\hline
Degree  & only $k$ & yes & yes\\
Harmonic & yes & yes & yes \\ 
Closeness & no & no & no \\
Lin & only $k$& no & no \\ 
Betweenness & only $p$ & no & no \\ 
Dominant & only $k$ & yes & no \\
Seeley & no & yes & no \\
Katz  & only $k$ & yes & yes\\
PageRank & no & yes & yes  \\
HITS & only $k$ &yes & no  \\
SALSA & no & yes & no \\
\end{tabular}
\caption{\label{tab:axiomatic} For each centrality and each axiom, we report whether it is satisfied.}
\end{table}

\section{Sanity check via information retrieval}

Information retrieval has developed in the last fifty years a large body of
research about extracting knowledge from data. In this section, we want to
leverage the work done in that field to check that our axioms
actually describe interesting features of centrality measures. We are in this sense
following the same line of thought as in~\cite{NZTHW}: in that paper, the
authors tried to establish in a methodologically sound way which of degree, HITS
and PageRank works better as a feature in web retrieval.
Here we ask the same question, but we include for the first time also geometric
indices, which had never been considered before in the literature about
information retrieval, most likely because it was not possible to compute them
efficiently on large networks.\footnote{It is actually now possible to approximate them efficiently~\cite{BoVHB}.}

The community working on information retrieval developed a number of standard datasets with 
associated queries and ground truth about which documents are relevant for every query; those
collections are typically used to compare the (de)merits of new retrieval methods.
Since many of those collections are made of hyperlinked documents, it is possible to use them 
to assess centrality measures, too. 

In this paper we consider the somewhat classical TREC GOV2
collection (about 25 million web documents) and the 149 associated queries. For each
query (\emph{topic}, in TREC parlance), we have solved the corresponding Boolean conjunction of terms, obtaining
a subset of matching web pages. Each subset induces a graph (whose nodes are the pages satisfying the conjunctive query), which can then be ranked using
any centrality measure. 
Finally, the pages in the graph are listed in score order
as results of the query, and standard relevance measures can be applied to see how much they correspond to the available 
ground truth about the assessed relevance of pages to queries.

There are a few methodological remarks that are necessary before discussing the
results:
\begin{itemize}
  \item The results we present are for GOV2; there are other publicly available
  collections with queries and relevant documents that can be used to this purpose.
  \item As observed in earlier works~\cite{NZTHW}, centrality scores in isolation have
  a very poor performance when compared with text-based ranking functions, but can improve
  the results of the latter. We purposely avoid measuring performance in conjunction with
  text-based ranking because this would introduce further parameters. Moreover, our idea is
  using information-retrieval techniques to judge centrality measures, not improving
  retrieval performance \emph{per se} (albeit, of course, a better centrality measure could be used to improve the quality of retrieved
  documents).
  \item Some methods are claimed to work better if \emph{nepotistic links} (that is, links between 
  pages of the same host) are excluded from the graph. Therefore, we report also results on
  the procedure applied to GOV2 with all intra-host links removed.
  \item There are several ways to build a graph associated with a query. Here, we choose a very
  straightforward approach---we solve the query in conjunctive form and build the induced subgraph. Variants may
  include enlarging the resulting graph with successors/predecessors, possibly by sampling~\cite{NGPLM}.
  \item There are many measures of effectiveness that are used in information retrieval; among those, we 
  focus here on the Precision at 10 (P@10, i.e., fraction of relevant documents retrieved among the first ten) and 
  on the NDCG@10~\cite{JKCGBEIRT}.
  \item Because of the poor performance, even for the best documents about half of the queries have null score. Thus,
  the data we report must be taken with a grain of salt---confidence intervals for our measures of effectiveness 
  would be largely overlapping (i.e., our experiments have limited statistical significance).
\end{itemize}

Our results are presented in Table~\ref{tab:ir}: even if obtained in a completely 
different way, they confirm the information we have been gathering with our axioms.
Harmonic centrality has the best overall scores. When we eliminate nepotistic links, the landscape
changes drastically---SALSA and PageRank now lead the results---but the best
performances are \emph{worse} than those obtained using the whole structure of the
web. Note that, again consistently with the information gathered up to now,
closeness performs very badly and betweenness performs essentially like using no
ranking at all (i.e., showing the documents in some arbitrary order).

There are four new centrality measures appearing in Table~\ref{tab:ir} that deserve an explanation.
When we first computed these tables, we were very puzzled: HITS is supposed to
work very badly on disconnected graphs (it fails score monotonicity), whereas it provided
the second best ranking after harmonic centrality. Also,
when one eliminates nepotistic links the graphs become highly disconnected and all rankings tend to
correlate with one another simply because most nodes obtain a null score. How is it possible that
PageRank and SALSA work so well (albeit less than harmonic centrality on
the whole graph) with so little information?


Our suspect was that \emph{these measures were actually picking up some much more
elementary signal than their definition could make one think}. In a highly
disconnected graph, the values assigned by such measures depend mainly on
the indegree and on some additional ranking provided by coreachable (or weakly
reachable) nodes.

We thus devised four ``naive'' centrality measures around our axioms. These measures
depend on density, and on size. We use two kinds of scores based on density:
the indegree and the \emph{negative $\beta$-measure}~\cite{BrGIBEO},\footnote{Note that the $\beta$-measure originally
defined by van den Brink and Gilles in~\cite{BrGIBEO} is the positive version, that is, the
negative $\beta$-measure can be obtained by applying the $\beta$-measure defined in~\cite{BrGIBEO} to the transposed graph.} 
 that is,
\[
\sum_{y\to x}\frac1{d^+(y)}.
\]
The negative $\beta$-measure is a kind of ``Markovian indegree'' inspired by the $\ell_1$ normalization typical of Seeley's index,
PageRank, and SALSA. Indeed, indegree is expressible in matrix form as $\mathbf 1 A$, whereas
the negative $\beta$-measure is expressible as $\mathbf 1 \bar A$.
We also use two kinds of scores based on size: the number of coreachable
nodes, and the number of weakly reachable nodes. By multiplying a score based on density with a score based on size, we obtain
four centrality measures, displayed in Table~\ref{tab:naive}, which satisfy all our axioms.

\begin{table}
\renewcommand{\arraystretch}{1.5}
\centering
\begin{tabular}{l|c|c}
	           & Indegree & Negative $\beta$-measure \\
\hline
Number of coreachable nodes & Indegree$^\leftarrow$ & $\beta$-measure$^\leftarrow$ \\
Number of weakly reachable nodes & Indegree$^\leftrightarrow$ & $\beta$-measure$^\leftrightarrow$ \\
\end{tabular}
\caption{\label{tab:naive}The names and definition of the four naive centrality measures used in Table~\ref{tab:ir}. Each
centrality is obtained by multiplying the values described by its row and column labels.}
\end{table}

As it is evident from Table~\ref{tab:ir}, such simple measures outperform in
this test most of the very sophisticated alternatives proposed in the literature:
this shows, on one hand, that it is possible to extract information from the graph
underlying a query in very simple ways that do not involve any spectral or geometric technique
and, on the other hand, that designing centralities around our axioms actually pays off.
We consider this fact as a further confirmation that the traits of centrality represented
by our axioms are important.

\begin{table}
\centering
\begin{tabular}{l|r|r}
\multicolumn{3}{c}{All links} \\
\hline
& \multicolumn{1}{c|}{NDCG@10} & \multicolumn{1}{c}{P@10} \\
\hline
BM25 & 0.5842 & 0.5644 \\
\hline
Harmonic & 0.1438 & 0.1416 \\
Indegree${^\leftarrow}$ & 0.1373 & 0.1356 \\
HITS & 0.1364 & 0.1349 \\
Indegree${^\leftrightarrow}$ & 0.1357 & 0.1349 \\
Lin & 0.1307 & 0.1289 \\
$\beta$-measure${^\leftrightarrow}$ & 0.1302 & 0.1322\\
$\beta$-measure${^\leftarrow}$ & 0.1275 & 0.1248 \\
Katz $3/4\lambda$ & 0.1228 & 0.1242 \\
Katz $1/2\lambda$ & 0.1222 & 0.1228 \\
Indegree & 0.1222 & 0.1208 \\
Katz $1/4\lambda$ & 0.1204 & 0.1181 \\
SALSA & 0.1194 & 0.1221 \\
Closeness & 0.1093 & 0.1114 \\
PageRank $1/2$ & 0.1091 & 0.1094 \\
PageRank $1/4$ & 0.1085 & 0.1107 \\
Dominant & 0.1061 & 0.1027 \\
PageRank $3/4$ & 0.1060 & 0.1094 \\
Betweenness & 0.0595 & 0.0584 \\
\hline
--- & 0.0588 & 0.0577 \\
\end{tabular}\qquad\qquad
\begin{tabular}{l|r|r}
\multicolumn{3}{c}{Inter-host links only} \\
\hline
& \multicolumn{1}{c|}{NDCG@10} & \multicolumn{1}{c}{P@10} \\
\hline
BM25 & 0.5842 & 0.5644 \\
\hline
$\beta$-measure${^\leftrightarrow}$ & 0.1417 & 0.1349\\
SALSA & 0.1384 & 0.1282\\
PageRank $1/4$ & 0.1347 & 0.1295\\
$\beta$-measure${^\leftarrow}$ & 0.1328 & 0.1275 \\
Indegree${^\leftrightarrow}$ & 0.1318 & 0.1255\\
PageRank $1/2$ & 0.1315 & 0.1268\\
PageRank $3/4$ & 0.1313 & 0.1255\\
Katz $1/2\lambda$ & 0.1297 & 0.1262\\
Indegree${^\leftarrow}$ & 0.1295 & 0.1262\\
Harmonic & 0.1293 & 0.1262\\
Katz $1/4\lambda$ & 0.1289 & 0.1255\\
Lin & 0.1286 & 0.1248\\
Indegree & 0.1283 & 0.1248\\
Katz $3/4\lambda$ & 0.1278 & 0.1242\\
HITS & 0.1179 & 0.1107\\
Closeness & 0.1168 & 0.1121\\
Dominant & 0.1131 & 0.1067\\
Betweenness & 0.0588 & 0.0577\\
\hline
--- & 0.0588 & 0.0577 \\
\end{tabular}
\caption{\label{tab:ir}Normalized discounted cumulative gain (NDCG) and precision at 10 retrieved documents (P@10) for
the GOV2 collection using all links and  using only inter-host links. The tables include, for reference,
the results obtained using a state-of-the-art text ranking function, BM25, and a final line obtained by
applying no ranking function at all (documents are sorted by the document identifier).}
\end{table}

\section{Conclusions and future work}

We have presented a set of axioms that try to capture part of the intended behavior of
centrality measures. We have proved or disproved all our axioms for ten
classical centrality measures and for \emph{harmonic centrality}, a
variant to Bavelas's closeness that we have introduced in this paper. The
results are surprising and confirmed by some information-retrieval experiments: harmonic centrality
is a very simple measure providing a good notion of centrality. It is almost identical
to closeness centrality on undirected, connected networks but provides a sensible centrality
notion for arbitrary directed graphs.

There is of course a large measure of arbitrariness in the formulation of our
axioms: we believe that this is actually a \emph{feature}---building an
ecosystem of interesting axioms is just a healthy way of understanding
centrality better and less anecdotally. Promoting the growth of such an
ecosystem is one of the goals of this work.

As a final note, the experiments on information retrieval that we have reported are just the beginning. Testing with
different collections (and possibly with different ways of generating the graph
associated with a query) may lead to different results. Nonetheless, we believe we
have made the important point that \emph{geometric measures are relevant not
only to social networks, but also to information retrieval}. In the literature
comparing exogenous (i.e., link-based) rankings one can find different instances
of spectral measures and indegree, but up to now the venerable measures based
on distances have been neglected. We suggest that it is time to change this
attitude.

\section{Acknowledgments}

We thank David Gleich for useful pointers leading to the proof of the
score monotonicity of PageRank in the general case, and Edith Cohen for useful
discussions on the behavior of centrality indices.
Marco Rosa participated in the first phases of development of this paper.

\hyphenation{ Vi-gna Sa-ba-di-ni Kath-ryn Ker-n-i-ghan Krom-mes Lar-ra-bee
  Pat-rick Port-able Post-Script Pren-tice Rich-ard Richt-er Ro-bert Sha-mos
  Spring-er The-o-dore Uz-ga-lis }\hyphenation{ Vi-gna Sa-ba-di-ni Kath-ryn
  Ker-n-i-ghan Krom-mes Lar-ra-bee Pat-rick Port-able Post-Script Pren-tice
  Rich-ard Richt-er Ro-bert Sha-mos Spring-er The-o-dore Uz-ga-lis }

\end{document}